\pdfoutput=1 
\documentclass[%
nofootinbib,
superscriptaddress,
 preprint,
showpacs,preprintnumbers,
 amsmath,amssymb,
 aps,
 pra,
 longbibliography,
 floatfix,
 lengthcheck,%
]{revtex4-2}

\usepackage{Hamiltonian_fluid_analogies}

\usepackage{xcolor}
\hypersetup{
	colorlinks=true,
	linkcolor=blue!50!red,
	urlcolor=red!70!black
} 

\AtBeginDocument{%
	\newwrite\bibnotes
	\def\bibnotesext{Notes.bib}
	\immediate\openout\bibnotes=\jobname\bibnotesext
	\immediate\write\bibnotes{@CONTROL{REVTEX42Control}}
	\immediate\write\bibnotes{@CONTROL{%
			apsrev42Control,author="08",editor="0 ",pages="1",title="1",year="1"}}
	\if@filesw
	\immediate\write\@auxout{\string\citation{apsrev42Control}}%
	\fi
}%

\begin{document}

\newcommand{\skipl}{}
\newtheorem{theorem}{Theorem}

\title{
Joint quantum-classical Hamilton variation principle in the phase space
}

\author{Dmitry V. Zhdanov}
\email{dm.zhdanov@gmail.com}
\affiliation{
Tulane University, New Orleans, LA 70118, USA
}
\author{Denys I. Bondar}
\affiliation{
	Tulane University, New Orleans, LA 70118, USA
}
\date{\today}
\begin{abstract}
We show that the dynamics of a closed quantum system obeys the Hamilton variation principle. Even though quantum particles lack well-defined trajectories, their evolution in the Husimi representation can be treated as a flow of multidimensional probability fluid in the phase space. By introducing the classical counterpart of the Husimi representation in a close analogy to the Koopman-von Neumann theory, one can largely unify the formulations of classical and quantum dynamics. We prove that the motions of elementary parcels of both classical and quantum Husimi fluid obey the Hamilton variational principle, and the differences between associated action functionals stem from the differences between classical and quantum pure states. The Husimi action functionals are not unique and defined up to the Skodje flux gauge fixing [R. T. Skodje et al{.} Phys.~Rev.~A \textbf{40}, 2894 (1989)]. We demonstrate that the gauge choice can dramatically alter flux trajectories. Applications of the presented theory for constructing semiclassical approximations and hybrid classical-quantum theories are discussed.
\end{abstract}
\maketitle

\section{Introduction}
Variational principles are the cornerstones for both classical and quantum mechanics, but their content is different in these two theories. The classical Hamilton variation principle (see, e.g., \cite{BOOK-Landau}) provides a recipe for identifying the phase space trajectories of classical particles via solving an extremal problem. The quantum-mechanical Dirac-Frenkel variational principle \cite{1934-Frenkel} is similar in that it also enables finding the time-dependent wavefunction as the solution of an extremal problem. The Dirac-Frenkel principle found applications in computational chemistry and many-body physics (see, e.g., Refs.~\cite{2000-Beck,2018-Benedikter,2004-Lubich,2000-Militzer}). However, its utility is limited in comparison to the classical counterpart because the phase space variables are merely formal parameters, half of which are implicitly defined and are hard to compute. In contrast, the Feynman's path integral approach \cite{2005-Feynman}, which is closely related to the Schwinger quantum action principle \cite{BOOK-Milton,1996-Thoss,2021-Manjarres}, can be specialized in terms of conventional phase space variables and is useful in a plethora of applications, from chemistry to quantum gravitation \cite{1993-Voth,1994-Storey,1979-Gerry,1994-Storey,2009-Tempere,2008-Hamber}. However, this approach involves the trajectory-selecting variational analysis only at the level of the semiclassical stationary phase approximation. 

At first glance, the Hamilton variation principle cannot be formulated for quantum particles since they lack well-defined trajectories in the phase space. We will show that this obstacle is illusive. The resolution comes from the following known facts about classical and quantum dynamics, which, however, were never systematically analyzed together:
\begin{enumerate}
	\item Both classical and quantum mechanics support the exact hydrodynamic analogies, in which state evolution is represented as a flow of effective probability fluid in the phase space \cite{1989-Skodje}. Thus, both classical and quantum dynamics can be described as a transport of fluid parcels along well-defined trajectories.
	\item ``The [Hamilton] formalisms in each form of mechanics are constructed to describe only pure states'' \cite{1979-Shirokov}.   
\end{enumerate}

The value of our Hamilton theory behind quantum phase space trajectories in explaining the observed quantum phenomena is analogous to the value of microscopic theory of classical fluid in explaining its macroscopic behavior. In particular, the present work lays theoretical foundations behind semiempirical and variationally defined quantum phase space trajectories widely and successfully used for modeling chemical reaction dynamics since the seminal papers by Eric Heller \cite{1975-Heller,1981-Heller}. Furthermore, the proposed theory presents the mathematical machinery for analyzing an unexplored infinite family of alternative rigorous definitions of phase space trajectories. 

This paper is organized as follows. First, in Sec.~\ref{*SEC:Math*} we will set up the unified framework for the mathematical description of both quantum and classical systems based on classical statistical mechanics and the Husimi representation of quantum mechanics. We also introduce the notion of admissible trajectory variations, which will enable us to express both quantum and classical mechanics on the same footing. Then, in Sec.~\ref{*SEC:JHVP} we will postulate the quantum-classical Hamilton variation principle and derive the generalized Hamilton equations for both theories. The constructed theory exhibits gauge freedom akin to one found in classical electrodynamics. In Sec.~\ref{*SEC:EHVP} we show how this gauge freedom can be incorporated into the variation principle. The paper concludes with a brief summary and outlook.

\section{Mathematical preliminaries\label{*SEC:Math*}}
The prerequisite for formulating the generalized quantum-classical Hamilton variation principle is to express both dynamics in terms of phase space trajectories. In this section we set up a theoretical framework allowing to do so for a generic $N$-particle bosonic system. The associated classical phase space is $2N$-dimensional and consists of $N$ momenta $\pp{=}\{p_1,...,p_N\}$ and $N$ positions $\xx{=}\{x_1,...,x_N\}$. The phase space formulations of quantum mechanics are diverse \cite{1966-Cohen} and have been known since its early days. The Wigner representation \cite{1932-Wigner,BOOK-Zachos,1989-Cohen} is arguably the most familiar one. The Wigner representation adopts the language of classical statistical mechanics. The system state is described by the real-valued Wigner function $\WF(\pp,\xx)$ defined as 
\begin{gather}\label{mp.-Wigner_function}
\WF(\pp,\xx){=}\tfrac{1}{(2\pi\hbar)^{\dimensionality}}\inftyints\matel{\xx{+}\tfrac{\llambda}2}{\hat{\rho}}{\xx{-}\tfrac{\llambda}2}e^{{-}i\frac{\pp\cdot\llambda}{\hbar}}\diff^{\dimensionality}\llambda,
\end{gather}
where $\hat{\rho}$ is the system's density matrix and $\ket{\xx}$ is the position eigenstate. The Wigner function allows to calculate the mean value $\bar O$ of any physical quantity $O$ alike the classical statistical average:
\begin{gather}\label{mp.-Wigner_averaging}
\bar O{=}\midop{\qw{O}\WF}, 
\end{gather}
where $\midop{...}{=}\inftyints...\diff^N\pp\diff^N\xx$ denotes the integration over the entire $2N$-dimensional phase space. The Weyl symbol $\qw{O}{=}\qw{O}(\pp,\xx)$ entering Eq.~\eqref{mp.-Wigner_averaging}  coincides with the expression for the associated classical physical quantity $\cw{O}(\pp,\xx)$ in the special case $\qw{O}{=}{\qw{O}}'(\pp){+}{\qw{O}}''(\xx)$. Furthermore, the normalization $\midop{\WF(t)}{=}1$ of $\WF$ is preserved at all times, and hence $\WF$ resembles a classical probability distribution $\PD$. Nevertheless, $\PD$ differs from $\WF$ in that the latter can take negative values \cite{2018-Oliva}. Perhaps, the closest classical analogy to the dynamics of $\WF$ is the charge redistribution inside an isolated medium, where the total charge is conserved, but the numbers of free positive and negative charges may vary due to charge separation and recombination. The varying number of effective changes is a fatal drawback for introducing a quantum analog of classical action, because the latter assumes a fixed number of particles and deals with trajectories having no sources and sinks.

Fortunately, quantum mechanics offers a remarkably simple way to ``fix'' the Wigner function negativity via the Gaussian convolution with the kernel
\begin{gather}\label{mp.-Husimi_kernel}
\KernW_{\wwW_{\pp},\wwW_{\xx}}(\delta\pp,\delta\xx){=}
\textstyle{\prod_{n=1}^{\dimensionality}}
\frac{1}{2\pi\wW_{p_n}\wW_{x_n}}e^{{-}\frac{\delta p_n^2}{2\wW_{p_n}^2}{-}\frac{\delta x_n^2}{2\wW_{x_n}^2}},
\end{gather}
where $\wW_{p_n}$ and $\wW_{x_n}$ are arbitrary positive numbers, such that $\wW_{p_n}\wW_{x_n}{=}\tfrac{\hbar}2$.
The resulting state function
\begin{gather}
\HF(\pp,\xx){=}\KernWD_{\wwW_{\pp},\wwW_{\xx}}\WF(\pp,\xx){=}\inftyints\WF(\pp',\xx')\times\notag\\
\KernW_{\wwW_{\pp},\wwW_{\xx}}(\pp{-}\pp',\xx{-}\xx') \diff^N\pp'\diff^N\xx'\label{mp.-HF_definition}
\end{gather}
reduces to
\begin{gather}\label{mp.-HF_definition_pure}
\HF(\ppc,\xxc){=}\tfrac{\matel{\ppc,\xxc}{\hat\rho}{\ppc,\xxc}}{(2\pi\hbar)^N}{\geq}0,
\end{gather}
where $\ket{\ppc, \xxc}$ is the multidimensional squeezed coherent state localized at $\{\ppc, \xxc\}$:
\begin{align}\label{mp.-CS_definition}
\scpr{\xx}{\ppc,\xxc}{=} \textstyle{\prod_{n{=}1}^{\dimensionality}}\pi^{{-}\frac14}\wc_n^{{-}\frac12}
e^{-\frac{(x_n{-}\xc_{n})^2}{2\wc_n^2}{+}\frac i{\hbar}\pc_{n}(x_n{-}\xc_{n})}.
\end{align}
Everywhere non-negative $\HF(\pp,\xx)$ is known as the Husimi function \cite{1940-Husimi}. Since the Husimi function preserves the normalization $\midop{\HF(t)}{=}1$, it satisfies the formal requirements to be a probability distribution%
\footnote{The convolution operator $\KernWD_{\wwW_{\pp},\wwW_{\xx}}$ can be re-expressed in the differential form	\begin{gather}\label{mp.-HF_definition_diff}
	\KernWD_{\wwW_{\pp},\wwW_{\xx}}{=}\prod_{n=1}^{\dimensionality}\exp({\tfrac12\wW_{p_n}^2\tpder{^2}{p_n^2}{+}\tfrac12\wW_{x_n}^2\tpder{^2}{x_n^2}})
\end{gather} 
(see, e.g., Ref.~\cite{2003-Ulmer}). Hence, operator $\KernWD_{\wwW_{\pp},\wwW_{\xx}}$ is translation-invariant.}%
.
Hereafter, we assume fixed values for the width parameters $\wwW_{\pp}$ and $\wwW_{\xx}$ and omit the associated indices for brevity.

By analogy with the Wigner representation, it is convenient to introduce the Husimi symbols $\qh{O}=\KernWD\qw{O}$ of physical observables. For any three quantum operators $\hat A$, $\hat B$ and $\hat C$, such that $\hat C{=}\hat A\hat B$, the associated Weyl and Husimi symbols relate as
\begin{gather}\label{mp.-algerbra_qn}
\qw{C}{=}\qw{A}{\star}\qw{B},~~~
\qh{C}{=}\qh{A}\starH\qh{B},
\end{gather}
where 
\begin{gather}\label{mp.-Moyal_product}
\star{=}\exp\big(i\tfrac{\hbar}2\psnprod\big),~~\starH{=}\star\scprH,~~\psnprod{=}\pderl{\xx}{\cdot}\pderr{\pp}{-}\pderl{\pp}{\cdot}\pderr{\xx},\\
\label{mp.-Husimi_scalar_product}
\scprH{=}\exp\left({\sum_{n{=}1}^{\dimensionality}\left(\wW_{x_n}^2\pderl{x_n}\pderr{x_n}{+}\wW_{p_n}^2\pderl{p_n}\pderr{p_n}\right)}\right).
\end{gather}
The arrows in the above expressions indicate the directions of differentiation, e.g., $f(\pp,\xx)\pderl{\xx}g(\pp,\xx)=g(\pp,\xx)\pderr{\xx}f(\pp,\xx){=}\tpder{f(\pp,\xx)}{\xx}g(\pp,\xx)$ \cite{2010-Polkovnikov}. Operation $\psnprod$ represents the classical Poisson bracket. Operation $\star$ is called the Moyal product \cite{1946-Groenewold}%
\footnote{Both $\star$ and $\starH$ are associative non-commutative generalized multiplication operations, which do not explicitly depend on $\pp$ and $\xx$. The latter ensures the translation invariance of the corresponding algebra in the phase space.}%
.%

In order to put quantum and classical descriptions on the same footing, let us introduce a Husimi-like representation of classical mechanics, where system states and symbols of physical quantities are defined as $\HD{=}\KernWD\PD$ and $\ch{O}{=}\KernWD\cw{O}$. If three classical physical quantities $\cl{A}$, $\cl{B}$ and $\cl{C}$ are related as $\cl{C}{=}\cl{A}\cl{B}$, then 
\begin{gather}\label{mp.-algerbra_cl}
\ch{C}{=}\ch{A}\scprH\ch{B}{=}\ch{B}\scprH\ch{A}.
\end{gather}
The transformations $\WF{\to}\HF$ and $\PD{\to}\HD$ are formally invertible and are analogous to anti-blurring of Gaussian-blurred images. However, these inverse transformations are ill-defined.

The formula to calculate the mean value of observables is identical for the classical and quantum Husimi representations:
\begin{gather}
\bar O{=}\midop{\hs{O}{\scprH}\hs{Q}}.
\end{gather}
Both $\HF$ and $\HD$ are non-negative and remain normalized at all times, so their evolution resembles the flow of a compressible fluid. Hence, both the state functions should formally satisfy the continuity equation \cite{1989-Skodje,2013-Veronez}
\begin{gather}\label{mp.-continuity_equation_for_Q}
\tpder{}{t}\hs{Q}{=}{-}\tpder{}{\xx}{\cdot}(\dot{\xx}\hs{Q}){-}\tpder{}{\pp}{\cdot}(\dot{\pp}\hs{Q}),
\end{gather}
where $\dot{\xx}$ and $\dot{\pp}$ are generalized phase space velocities of elemental fluid parcels. It is worth to stress the difference between evolution of a classical Husimi fluid $\HD$ and a conventional probability distribution $\PD$. In the latter case, the phase space velocities $\dot{\xx}(\ppc,\xxc)$ and $\dot{\pp}(\ppc,\xxc)$ characterize kinematics of a physical particle $\PD{\propto}\delta_{\pp{-}\ppc(t),\xx{-}\xxc(t)}$ located at $\{\ppc,\xxc\}$. In the classical Husimi picture, this particle is represented by the Gaussian blob 
\begin{gather}\label{mp.-classical_pure_state}
\ch{Q}=\textstyle{\prod_{n=1}^{\dimensionality}}
\frac{1}{2\pi\wW_{p_n}\wW_{x_n}}\exp\left({{-}\tfrac{(\pc_n{-} p_n)^2}{2\wW_{p_n}^2}{-}\tfrac{(\xc_n{-}x_n)^2}{2\wW_{x_n}^2}}\right),
\end{gather}
and its time evolution is described by a bundle of trajectories covering the entire phase space. Each trajectory can be chosen as a shifted replica of the actual physical trajectory $\{\ppc(t),\xxc(t)\}$. Note, however, that parallel translations can be combined with arbitrary circulations of fluid parcels inside the blob. This freedom is the essence of the Skodje gauge invariance discussed in Sec.~\ref{*SEC:EHVP}. Furthermore, the classical Liouvillian in the Husimi picture is non-local and in that sense is like quantum Liouvillian. Consequently, when considering a statistical ensemble of two or more particles, the associated elemental fluid trajectories cease to replicate actual physical trajectories. 

Thus, the Husimi representation brings classical and quantum mechanics onto the common ground by enabling a consistent fluid description of quantum mechanics and complicating the description of classical dynamics with quantum-like artifacts. In this picture, point particles described by Eq.~\eqref{mp.-classical_pure_state} play the role of classical pure states. Quantum pure states can be identified in the Wigner and Husimi representations via the following relations:
\begin{subequations}\label{mp.-purity_definition}
\begin{gather}\label{mp.-purity_definition_cl}
(2\pi\hbar)^{N}\WF{\star}\WF{=}\WF,\\
(2\pi\hbar)^{N}\HF{\starH}\HF{=}\HF.\label{mp.-purity_definition_qn}
\end{gather}
\end{subequations}

Recall that the Hamilton variation principle is formulated for a pure classical state and considers a variation of a single physical trajectory. However, any time-dependent pure state in the Husimi picture is represented by a bundle of coupled trajectories, and an arbitrary single-trajectory variation would generally break the purity. We resolve these issues as follows:
\begin{itemize}
\item We look for a variation principle applicable to every trajectory in the bundle. 
\item We admit only collective trajectories variations, which keep the purity conserved and satisfy the following definition:
\end{itemize}

\begin{definition}\label{mp.-def_admissible_variation}
	An arbitrary time-dependent parallel translation $\{\delta\xx(t),\delta\xx(t)\}$ of all phase space points $\{\pp,\xx\}{\to}\{\pp{+}\delta\pp(t),\xx{+}\delta\xx(t)\}$ is called \emph{the admissible variation} of a system state $\hs{Q}(\pp,\xx,t)$.
\end{definition}

Apart from fluid circulations inside the blobs, the admissible variations are the only generic purity-preserving variations applicable for any pure (classical or quantum) state.

To facilitate the definition of the action for individual Husimi trajectories, it is convenient to rewrite the fluid analogy \eqref{mp.-continuity_equation_for_Q} in the Lagrangian picture by assigning a Lagrangian label $L$ to each fluid parcel \cite{1988-Salmon}. The specific form of the label is not important for our discussion, the only important requirement is that the label moves with the fluid and does not change in time for a specific fluid parcel. For instance, $L$ can be a $2\dimensionality$-dimensional real vector of the phase space coordinates of a fluid parcel at the initial time. Such labels would define a Lagrangian frame moving with the fluid. Regardless of their structure, the labels obey the evolution equation
\begin{gather}\label{mp.-labels-evolution}
\tpder{}{t}L(\pp,\xx,t){=}{-}\dot{\xx}{\cdot}\tpder{}{\xx}L(\pp,\xx,t){-}\dot{\pp}{\cdot}\tpder{}{\pp}L(\pp,\xx,t).
\end{gather}

One can check that the admissible variations satisfy the following property, which holds for an arbitrary classical or quantum state $\hs{Q}$, an arbitrary phase space operator $\kvnO(\pp,\xx,\pder{}{\pp},\pder{}{\xx})$ and an arbitrary fluid trajectories labeling $L$:
\begin{gather}\label{mp.-matel_var}
\delta\midop{L\kvnO\hs{Q}}{=}\midop{L(\delta\xx{\cdot}\tpder{\kvnO}\xx{+}\delta\pp{\cdot}\tpder{\kvnO}{\pp})\hs{Q}}.
\end{gather}

The gauge invariance analysis in Sec.~\ref{*SEC:EHVP} will additionally require us to deal with time-dependent properties of trajectories. We will call such properties as features $F(\pp,\xx,t)$. By definition, features also obey Eq.~\eqref{mp.-matel_var}, but their evolution is described by the equation
\begin{gather}\label{mp.-features-evolution}
\tpder{}{t}F{=}{-}\dot{\xx}{\cdot}\tpder{}{\xx}F{-}\dot{\pp}{\cdot}\tpder{}{\pp}F{+}\tder{}{t}F.
\end{gather}

\section{Joint quantum-classical Hamilton variation principle\label{*SEC:JHVP}}
Consider an arbitrary (not necessirely unique) Lagrangian labeling $L$ of effective fluid parcels representing the state $\hs{Q}(t)$ of a closed classical or quantum system. Let us introduce the \emph{Hamilton generator} operator $\hs{\hamgen}{=}\hs{\hamgen}(\pp,\xx,\tpder{}{\pp},\tpder{}{\xx})$ defined as
\begin{subequations}\label{hp.-Husimi_ham_gen}
\begin{align}
&\qh{\hamgen}{=}\qh{H}(\pp,\xx)\sinc\left(\tfrac{\hbar}2\psnprod\right),&\mbox{(quantum mechanics)}\label{hp.-Husimi_ham_gen_q}\\
&\ch{\hamgen}{=}\ch{H}(\pp,\xx),&\mbox{(classical mechanics)}\label{hp.-Husimi_ham_gen_c}
\end{align}
\end{subequations}
where $\qh{H}$ and $\ch{H}$ are the Husimi symbols of the quantum and classical Hamiltonians, respectively, and $\sinc(z){=}\tfrac{\sin(z)}z$.

\begin{theorem}[Joint Hamilton variation principle]\label{*theorem-jhvp}
The evolution of both classical and quantum states $\hs{Q}$ satisfies the Hamilton variation principle
\begin{gather}\label{hp.-JHVP}
\delta S{=}\delta\timeint\diff t\big(\midop{L(\pp{\cdot}\dot{\xx})\hs{Q}}{-}\midop{L(\hs{\hamgen}{\scprH}\hs{Q})}\big){=}0,
\end{gather}
relative to all admissible (in the sense of Definition~\ref{mp.-def_admissible_variation}) variations $\delta S$ of action $S$ fulfilling the boundary conditions 
\begin{gather}\label{hp.-boundary_conditions_on_var}
\delta\xx(t_{\mathrm i}){=}\delta\xx(t_{\mathrm f}){=}0,~~ \delta\pp(t_{\mathrm i}){=}\delta\pp(t_{\mathrm f}){=}0
\end{gather}
at the initial and final times $t_{\mathrm i}$ and $t_{\mathrm f}$.
\end{theorem}
\begin{proof}
Consider two independent admissible variations $\pp{\to}\pp{+}\delta\pp(t)$ and $\xx{\to}\xx{+}\delta\xx(t)$ of fluid trajectories in Eq.~\eqref{hp.-JHVP}. With the help of identity \eqref{mp.-matel_var}, one gets\begin{subequations}
\label{hp.-la_vars}\begin{gather}
	\label{hp.-la_vars_p}
	\timeint\diff t~\delta\pp{\cdot}\midop{L(\dot{\xx}\hs{Q}{-}\tpder{\hs{\hamgen}}{\pp}{\scprH}\hs{Q})}{=}0,\\
	\label{hp.-la_vars_x}
	\timeint\diff t~ \delta\xx{\cdot}\midop{L({-}\dot{\pp}\hs{Q}{-}\tpder{\hs{\hamgen}}{\xx}{\scprH}\hs{Q})}{=}0.
	\end{gather} 
\end{subequations}
To obtain Eq.~\eqref{hp.-la_vars_x}, we also used the relation \begin{gather}\delta\timeint\diff t\midop{L(\pp{\cdot}\dot{\xx})\hs{Q}}{=}\delta\timeint\diff t\midop{L(\tder{(\pp{\cdot}\xx)}{t}-\dot{\pp}{\cdot}\xx)\hs{Q}}=\notag\\
{-}\delta\timeint\diff t\midop{L(\dot{\pp}{\cdot}\xx)\hs{Q}},
\end{gather}
where Eqs.~\eqref{mp.-continuity_equation_for_Q}, \eqref{mp.-labels-evolution} and the boundary conditions \eqref{hp.-boundary_conditions_on_var} are utilized.

Since $L(\pp,\xx)$, $\delta\xx(t)$ and $\delta\pp(t)$ are arbitrary functions, Eqs.~\eqref{hp.-la_vars} can be satisfied only if
\begin{subequations}\label{hp.-ham_eqs}
	\begin{gather}\label{hp.-ham_eq_p}
	\dot{\pp}\hs{Q}(\pp,\xx){=}{-}\tpder{\hs{\hamgen}}{\xx}{\scprH}\hs{Q}(\pp,\xx),\\
	\label{hp.-ham_eq_x}
	\dot{\xx}\hs{Q}(\pp,\xx){=}\tpder{\hs{\hamgen}}{\pp}{\scprH}\hs{Q}(\pp,\xx).
	\end{gather}
\end{subequations}
We call Eqs.~\eqref{hp.-ham_eqs} the generalized Hamilton equations. Their substitution into the continuity equation \eqref{mp.-continuity_equation_for_Q} gives after little algebra 
\begin{gather}\label{hs.-gen_Liouville-Husimi_equation}
\tpder{}{t}\hs{Q}{=}\hs{\hamgen}{\psnprod}{\scprH}\hs{Q},
\end{gather}
where $\psnprod$ is the classical Poisson bracket operation defined by Eq.~\eqref{mp.-Moyal_product}. Using the definitions \eqref{hp.-Husimi_ham_gen} of the Hamilton generators for classical and quantum mechanics, Eq.~\eqref{hs.-gen_Liouville-Husimi_equation} expands into
\begin{subequations}\label{hs.-gen_Liouville-Husimi_equation_exp}
\begin{gather}\label{hs.-gen_Liouville-Husimi_equation_exp_qn}
\tpder{}{t}\qh{Q}{=}{-}\tfrac{i}{\hbar}\left(\qh{H}{\starH}\qh{Q}{-}\qh{Q}{\starH}\qh{H}\right),\\
\tpder{}{t}\ch{Q}{=}\ch{H}{\psnprod}{\scprH}\ch{Q}.
\end{gather}
\end{subequations}
\begin{subequations}\label{hs.-gen_Liouville_equation}
Application of the inverse Husimi transformation to both sides of these equations gives the familiar classical Liouville equation 
\begin{gather}\label{hs.-gen_Liouville_equation_cl}
\tpder{}{t}\cw{P}{=}\cw{H}{\psnprod}\cw{P},
\end{gather}
and the quantum Liouville equation in the Wigner representation
\begin{gather}\label{hs.-gen_Liouville_equation_qn}
\tpder{}{t}\qw{P}{=}{-}\tfrac{i}{\hbar}\left(\qw{H}{\star}\qw{P}{-}\qw{P}{\star}\qw{H}\right).
\end{gather}
\end{subequations}
This finishes the proof.
\end{proof}

\section{Extended Hamilton variation principle. Skodje gauge freedom\label{*SEC:EHVP}}
We mentioned in Sec.~\ref{*SEC:Math*} the analogy between the dynamics of a quantum state $\wg{P}$ and an electric charge distribution. According to Noether's theorem, the charge conservation originates from the gauge invariance of electromagnetic potentials. Similarly, the hydrodynamic phase space velocities $\dot{\pp}$ and $\dot{\xx}$ in Eq.~\eqref{mp.-continuity_equation_for_Q} are defined up to arbitrary divergenceless fluxes, which do not alter the local density of effective fluid $\hs{Q}$. This fact will be hereafter referred as the Skodje flux gauge invariance to pay tribute to the seminal paper \cite{1989-Skodje}. In this section, we generalize the definition of the action $S$ in the Husimi representation to account for the Skodje gauge freedom.

We will need the following additional notations. Let  $\zz{=}\{\xx,\pp\}^{\intercal}$, $\tilde{\zz}{=}\Omega\zz{=}\{\pp,{-}\xx\}^{\intercal}$ and $\dot{\zz}{=}\{\dot{\xx},\dot{\pp}\}^{\intercal}$ be two $2N$-dimensional vectors of the phase space coordinates and velocities. Here
\begin{gather}\label{hp.-symplectic_matrix}
\Omega{=}\begin{pmatrix}0&I_N\\-I_N&0\end{pmatrix},
\end{gather}
is symplectic matrix, where $I_N$ is the $N{\times}N$ identity matrix. Denote by $\GFC_{i,j}(\pp,\xx,t)$ $(i,j{=}1,...,2N)$ the $2N{\times}2N$ set of arbitrary real-valued features (introduced at the end of Sec.~\ref{*SEC:Math*}), such that $\GFC_{i,j}{=}{-}\GFC_{j,i}$.

\begin{theorem}[The extended Hamilton variation principle]\label{*theorem-ehvp}
The statement of Theorem~\ref{*theorem-jhvp} also applies to the extended action
\begin{gather}
S{=}\timeint\diff t\big(\midop{L(\pp{\cdot}\dot{\xx})\hs{Q}}{-}\midop{L(\hs{\hamgen}{\scprH}\hs{Q})}{+}\midop{L(\Acal{+}\tilde{\zz}{\cdot}\LLcal)}\big),\tag{\ref*{hp.-JHVP}*}\label{hp.-EHVP}
\end{gather}
where $\Acal(\pp,\xx,t)$ is the Skodje gauge potential, which is an arbitrary twice differentiable function of its arguments, and the $2N$-dimensional vector $\LLcal$ with elements $\Lcal_i{=}\sum_{j{=}1}^{2N}\tpder{\GFC_{i,j}}{z_j}$ plays the role of the Skodje vector gauge.
\end{theorem}
\begin{proof}
Repeating the steps leading to Eqs.~\eqref{hp.-ham_eqs} in the proof of Theorem~\ref{*theorem-jhvp} one obtains

	\begin{gather}\label{hp.-ext_ham_eqs}
	\dot{\zz}\hs{Q}(\pp,\xx){=}\tpder{\hs{\hamgen}}{\tilde{\zz}}{\scprH}\hs{Q}(\pp,\xx){+}\tpder{A}{\tilde{\zz}}{+}\LLcal.
	\end{gather}

These extended Hamilton equations of motion differ from Eqs.~\eqref{hp.-ham_eqs} by the additional last terms containing the scalar and vector Skodje gauges $\Acal$ and $\LLcal$. However, when substituting them into the continuity equation \eqref{mp.-continuity_equation_for_Q}, all the terms containing the gauge potentials cancel out leading to Eqs.~\eqref{hs.-gen_Liouville_equation}, which finishes the proof.
\end{proof}

Note that the gauges $\Acal$ and $\LLcal$ also admit for their own ``internal'' gauge freedom in the sense that for any twice-differentiable real function $f(\pp,\xx,t)$ obeying Eq.~\eqref{mp.-features-evolution} the transformation
\begin{gather}
\Acal\to\Acal{+}f,~~\GFC_{i,i+N}{\to}\GFC_{i,i+N}{-}f,\notag\\
\GFC_{i+N,i}{\to}\GFC_{i+N,i}{+}f~~(i{=}1,...,N)
\end{gather}
does not change the right-hand side of Eq.~\eqref{hp.-ext_ham_eqs}. 

From the perspective of differential geometry, the gauge terms $\GFC_{i,j}$ can be seen as the coefficients of the $(2N{-}2)$-form
\begin{gather}
{\mathfrak T}{=}\sum_{i,j}\GFC_{i,j}{^*T^{i,j}}{=}2\sum_{i>j}\GFC_{i,j}{^*T^{i,j}},
\end{gather} 
where $T^{i,j}$ denotes a 2-vector with components along $i$th and $j$th phase space dimensions and $^*T^{i,j}$ is the $(2N-2)$-form dual to $T^{i,j}$. The Skodje gauge freedom described by theorem~\ref{*theorem-ehvp} follows from the exactness of closed forms $\tilde{\diff}{\mathfrak T}{=}\sum_i\LLcal_i~{^*(\tpder{}{z_i}})$ and $\tilde{\diff}{A}$
\begin{gather}\label{hp.-EHVP_geometric_meaning}
\tilde{\diff}(\tilde{\diff}{\mathfrak T})=0,~~\tilde{\diff}(\tilde{\diff}{A})=0,
\end{gather}
where $\tilde{\diff}$ denotes the external derivation \cite{BOOK-Schutz}.

It is instructive to consider two examples of the Skodje gauge transformations of Husimi flows. Denote by $\kkappa_i$ $(i{=}1,...,2N)$ the set of constant vectors (of an arbitrary length) pointing along each phase space dimension and construct the vector
$\kk(\aa){=}\sum_i a_i\kkappa_i$, where $a_i$ are some real coefficients. Let us also denote $\JJ(\zz){=}\tpder{\hs{\hamgen}}{\tilde{\zz}}{\scprH}\hs{Q}(\pp,\xx)$ for brevity. Consider the following gauges:
\begin{align}
\GFC_{i,j}{=}&\int_0^1\diff\alpha~{\alpha^{2N{-}2}}\big(z_iJ_j(\alpha\zz){-}z_jJ_i(\alpha\zz)\big),\label{hp.-sample_gauge_a}\\
\GFC_{i,j}{=}&\frac1{S_{2N}}\int_0^{\infty}\diff\alpha\int\diff^{2N-1}\SolidAngle\big\{k_j(\SolidAngle)J_i(\zz{+}\alpha\kk(\SolidAngle)){-}\notag\\
&k_i(\SolidAngle)J_j(\zz{+}\alpha\kk(\SolidAngle))\big\},\label{hp.-sample_gauge_b}
\end{align}
where $S_{2N}{=}\frac{{2\pi}^N}{\Gamma(N)}$ is the area of the unit sphere in an $2N$-dimensional space and integration $\int\diff^{2N-1}\SolidAngle$ is performed over the surface of the unit sphere $|\aa|{=}1$. Straightforward algebra shows that the gauge \eqref{hp.-sample_gauge_a} converts Eq.~\eqref{hp.-ext_ham_eqs} into
\begin{gather}\label{hp.-ext_ham_eqs_sample_a}
\dot{\zz}\hs{Q}(\pp,\xx){=}\int_0^1\diff\alpha~\alpha^{2N-1}\zz~\mathrm{div}\JJ(\alpha\zz).
\end{gather}
The respective equation for the gauge \eqref{hp.-sample_gauge_b} reads
\begin{align}
\dot{\zz}\hs{Q}(\pp,\xx){=}&{-}\tfrac1{S_{2N}}\int_0^{\infty}\diff\alpha\int\diff^{2N-1}\SolidAngle\notag\\
&\kk(\SolidAngle)\mathrm{div}\JJ(\zz{+}\alpha\kk(\SolidAngle)).
\label{hp.-ext_ham_eqs_sample_b}
\end{align}
Here we used notation $\mathrm{div}\JJ(\alpha\zz){=}\sum_{i{=}1}^{2N}\pder{J_i(\alpha\zz)}{z_i}$.
In the special case when all $\kkappa_i$ are unit vectors, Eq.~\eqref{hp.-ext_ham_eqs_sample_b} can also be rewritten as 
\begin{gather}
\dot{\zz}\hs{Q}(\pp,\xx){=}{-}\tpder{\tilde U}{\zz},\notag\\
\tilde U{=}\tfrac1{(2N{-}2)S_{2N}}\inftyints\diff^n\zz'\tfrac{\mathrm{div}\JJ(\zz)}{|\zz-\zz'|^{2N-2}}.
\end{gather} 

The right hand sides of Eqs.~\eqref{hp.-ext_ham_eqs_sample_a} and \eqref{hp.-ext_ham_eqs_sample_b} vanish for any steady state. In contrast, the right hand sides of Eqs.~\eqref{hp.-ham_eqs}, representing the case of the ``default'' gauge $\GFC{=}A{=}0$, vanish for only classical (but not for quantum) steady pure states. Thus, the significant difference between the two mechanics appears to be a mere consequence of the gauge convention. Note that vanishing phase space dynamics for steady states is one of the hallmarks of the Madelung theory (Bohmian mechanics). In fact, we have shown in Ref.~\cite{2020-Zhdanov} that the Bohmian mechanics is just another Skodje-gauge-transformed Husimi representation. 

\section{Discussion and conclusion\label{*SEC:DISS}}
In this paper, we extended the Husimi representation of quantum mechanics as the equivalent Hamilton theory of a classical many-body system with non-local interactions, where trajectories of ``particles'' obey an analog of the classical least action principle. The trajectory-based description makes our theory distinct from other phase space formulations of quantum mechanics%
\footnote{Madelung-de Broglie theory (Bohmian mechanics), which also enables trajectory-based description of quantum processes, appears to be a special case of the presented theory restricted to pure states, as shown in Ref.~\cite{2020-Zhdanov},}%
, and puts it in the form as similar to the Hamiltonian mechanics as possible.

Our theory extends the Wigner representation, which is simultaneously the phase space as well as the Hilbert space theory. This very duality (briefly reviewed in Appendix~\ref{*APP01}) enabled the classical-like reformulation of quantum mechanics developed in this paper. Note that one can also go the other way and derive quantum-like reformulation of classical mechanics; the result is known as the Koopman-von Neumann (KvN) theory. The link between the KvN theory and the Wigner representation is explained in Appendix~\ref{*APP01} and Ref.~\cite{2012-Bondar}.

Our variation principle formulation, Eq.~\eqref{hp.-EHVP}, and the KvN theory highlight the fact that the difference between classical and quantum mechanics stems exclusively from the different definitions of pure states. As discussed before Eq.~\eqref{mp.-purity_definition}, pure states in classical mechanics represent point particles described by the $\delta$-functions in the KvN formulation. To  preserve purity, the KvN generator of motion must be a first-order differential operator. The only skew-Hermitian operator satisfying this condition is the Poisson bracket $\cw{H}{\psnprod}$. The quantum-mechanical definition \eqref{mp.-purity_definition} implies that the purity-preserving variations $\delta\qh{Q}$ of $\qh{Q}$ satisfy the equality
\begin{gather}\label{03.-purity_variation_condition}
\delta\qh{Q}{\starH}\qh{Q}{+}\qh{Q}{\starH}\delta\qh{Q}{=}\delta \qh{Q}(2\pi\hbar)^{{-}N}.
\end{gather}
The equality \eqref{03.-purity_variation_condition} matches the general definition of derivation (see Ref.~\cite{BOOK-Nazaikinskii}, Sec.~3, Lemma I.1), implying that $\delta\qh{Q}{\propto}{-}i\left(\qh{H}{\starH}\qh{Q}{-}\qh{Q}{\star}\qh{H}\right)$ for some $\qh{H}(\pp,\xx)$, which matches \eqref{hs.-gen_Liouville-Husimi_equation_exp_qn}.

Despite our theory is formulated for pure states, the resulting Hamilton equations of motion \eqref{hp.-ext_ham_eqs}, as well as the master equations \eqref{hs.-gen_Liouville-Husimi_equation_exp} and \eqref{hs.-gen_Liouville_equation} are linear in the system state $\hs{Q}(\pp,\xx)$ and hence are valid for both pure and mixed states.

Note that the limit $\hbar{\to}0$ in Eq.~\eqref{mp.-purity_definition_qn} is singular, which clearly shows that classical mechanics is an approximation, which is well-defined for small, but still finite values of $\hbar$ supporting the $\starH$-operation with the property \eqref{mp.-purity_definition_qn}. In other words, even in classical mechanics, the dynamics ``inside'' the pure states defined by theorem~\ref{*theorem-jhvp} is non-trivial, but the resolution of this fact is hindered by their collapse to $\delta$-like phase space probability distributions. 

Theorem~\ref{*theorem-jhvp} seems to suggest that, unlike in classical mechanics, the quantum probability fluid does not get still even in the case of stationary pure states. However, theorem~\ref{*theorem-ehvp} shows that this distinction is mere consequence of a specific Skodje gauge choice. We explicitly show two Skodje gauges removing circulation fluxes from arbitrary steady states, the third example (corresponding to Bohmian mechanics) can be found in Ref.~\cite{2020-Zhdanov}. In the same way, one can construct infinitely many other fluid analogies with strikingly different fluxes. Thus, the Skodje gauge freedom embedded into the joint Hamilton variation principle is an unexplored resource for improving modern semiclassical approximations to wavepackets dynamics, such as developed in \cite{2013-Ohsawa}, and advancing numerical methods used in quantum chemistry \cite{2020-Zhdanov,2008-Shalashilin,2012-Hughes,2012-Wang} and many-body physics \cite{2017-Foss-Feig}. It would be also interesting to consider the joint Hamilton variation principle as a starting point for developing quantum-classical hybrid models alike recently proposed in Ref.~\cite{2019-Bondar}. Last but not least, we believe that our results are of a pedagogical interest to remove the conceptual barriers between the classical and quantum worlds.

\acknowledgements The authors thank Francisco Gonzalez Montoya and Tamar Seideman for valuable discussions and useful suggestions on improving the manuscript. D.I.B. is supported by Army Research Office (ARO) (grant W911NF-19-1-0377, program manager Dr.~James Joseph, and cooperative agreement W911NF-21-2-0139).

\appendix
\section{Wigner representation of quantum mechanics as the Hilbert space theory\label{*APP01}}
Here we outline the pedagogical way to reveal the Hilbert space structure underlying the Wigner phase space representation of quantum mechanics with the help of Bopp operators \cite{1956-Bopp}. Left $\lBopp{O}$, and right $\rBopp{O}$ Bopp operators associated with an arbitrary Weyl symbol $\qw{O}$ are defined via relations
\begin{gather}
\lBopp{O}\WF{=}\qw{O}{\star}\WF,~~\rBopp{O}\WF{=}\WF{\star}\qw{O},
\end{gather}
which should be valid for an arbitrary Wigner function $\WF$. Specifically, the Bopp operators associated with canonical position and momentum coordinates read
\begin{subequations}\label{app.-px_Bopp_operators}
\begin{align}
\lBopp{p}_n&{=}p_n{-}\tfrac{i\hbar}2\tpder{}{x_n},&\lBopp{x}_n&{=}x_n{+}\tfrac{i\hbar}2\tpder{}{p_n}, \label{app.-px_Bopp_operators_l}\\
\rBopp{p}_n&{=}p_n{+}\tfrac{i\hbar}2\tpder{}{x_n},&\rBopp{x}_n&{=}x_n{-}\tfrac{i\hbar}2\tpder{}{p_n}. \label{app.-px_Bopp_operators_r}
\end{align}
\end{subequations}

Left Bopp operators \eqref{app.-px_Bopp_operators_l} obey commutation relations similar to canonical commutation relations between position and momentum operators: $[\lBopp{p}_n,\lBopp{x}_n]{=}{-}i\hbar$, $[\lBopp{p}_n,\lBopp{p}_m]{=}0$ etc. Right Bopp operators \eqref{app.-px_Bopp_operators_r} satisfy the alike relations, which only differ by complex conjugation, e.g., $[\rBopp{p}_n,\rBopp{x}_n]{=}i\hbar$. Furthermore, left and right Bopp operators commute. Because of these remarkable properties, left (right) Bopp operators associated with any quantum observable $\hat O{=}O(\hat{\pp}, \hat{\xx})$ can be found simply by replacing $\hat{\pp}$ and $\hat{\xx}$ with $\lBopp{\pp}$ and $\lBopp{\xx}$ ($\rBopp{\pp}$ and $\rBopp{\xx}$) in $O(\hat{\pp}, \hat{\xx})$. 

Note that Bopp operators \eqref{app.-px_Bopp_operators} resemble linear combinations of coordinate and momentum operators with real coefficients for an effective system of $2N$ particles, where one half of operators is written in the coordinate representation and another half in the momentum representation. Hence, any Bopp operator can also be interpreted as a Hermitian operator describing this effective system. The corresponding $2N$-dimensional Hilbert space is exactly the Hilbert space of Wigner representation. 

One can further leverage the analogy with the effective system and define ``wavefunctions'' $\Psi(\pp,\xx)$ on the just-defined Hilbert space, which inner product is $\scpr{\Psi_1}{\Psi_2}{=}\inftyints\diff\pp\diff\xx\Psi_1^*\Psi_2$. It is possible to show that in the case of quantum pure states the Wigner function $\WF$ coincides with $\Psi(\pp,\xx)$ up to scaling factor ($\Psi(\pp,\xx){=}(2\pi\hbar)^{\dimensionality/2}\WF$) \cite{2013-Bondar}. The Bopp representation
\begin{gather}\label{app.-quantum_Liouville_equation}
\tpder{}{t}\Psi{=}{-}\tfrac{i}{\hbar}(\lBopp{H}{-}\rBopp{H})\Psi
\end{gather}
of quantum Liouville equation \eqref{hs.-gen_Liouville_equation_qn} has the mathematical structure of Schr\"odinger equation for the effective system, where the operator $\lBopp{H}{-}\rBopp{H}$ plays the role of effective Hamiltonian. However, in contrast to ordinary Schr\"odinger equation, the classical limit of the generator of motion in Eq.~\eqref{app.-quantum_Liouville_equation} is well-defined. Namely,  $\left.{-}\tfrac{i}{\hbar}(\lBopp{H}-\rBopp{H})\right|_{\hbar{\to}0}{=}\cw{H}{\psnprod}$, the classical Poisson bracket, which is the generator of motion in the KvN Hilbert space formulation of classical mechanics. Respectively, the classical limit $\left.\Psi\right|_{\hbar{\to}0}$ of effective wavefunctions in \eqref{app.-quantum_Liouville_equation} is the KvN wavefunction of classical state \cite{2012-Bondar}.

\bibliography{Hamiltonian_fluid_analogies}

\end{document}